\documentclass[oribibl]{llncs}
\usepackage{amsmath,amsfonts, times}
\usepackage[small,it]{caption}

\usepackage{hyperref}

\usepackage{graphicx}

\usepackage{algorithm}
\usepackage{algorithmicx}
\usepackage{algpseudocode}

\newcommand{\littlesum}{\mathop{\textstyle \sum}}

\newcommand{\sw}{\ensuremath\mathrm{SW}}
\newcommand{\rev}{\ensuremath\mathrm{Rev}}
\newcommand{\Exp}{\ensuremath\mathbb{E}}

\newcommand{\eps}{\epsilon}    

\newcommand{\R}{\mathbb{R}} 

\title{Efficiency-Revenue Trade-offs in Auctions}
\author{Ilias Diakonikolas\inst{1}\thanks{Supported by a Simons Postdoctoral Fellowship.} \and Christos Papadimitriou\inst{1}\thanks{Research supported by NSF grant CC-0964033 and by a Google University Research Award.} \and George Pierrakos\inst{1}$^\ast$$^\ast$ \and Yaron Singer\inst{2}\thanks{Research supported by a Microsoft Graduate Fellowship and a Facebook Graduate Fellowship.}}
\institute{UC Berkeley, EECS, \email{\{ilias,christos,georgios\}@cs.berkeley.edu}
\and
Google, Inc., \email{yaron@cs.berkeley.edu}}

\begin{document}
\maketitle

\begin{abstract}
When agents with independent priors bid for a single item, Myerson's optimal auction maximizes expected revenue, whereas Vickrey's second-price auction optimizes social welfare.   We address the natural question of  {\em trade-offs} between the two criteria, that is, auctions that optimize, say, revenue under the constraint that the welfare is above a given level.  If one allows for randomized mechanisms, it is easy to see that there are polynomial-time mechanisms that achieve any point in the trade-off (the {\em Pareto curve\/}) between revenue and welfare.  We investigate whether one can achieve the same guarantees using  {\em deterministic} mechanisms.  We provide a negative answer to this question by showing that this is a (weakly) NP-hard problem. On the positive side, we provide polynomial-time deterministic mechanisms that approximate with arbitrary precision any point of the trade-off between these two fundamental objectives for the case of two bidders, even when the valuations are correlated arbitrarily.  The major problem left open by our work is whether there is such an algorithm for three or more bidders with independent valuation distributions.
\end{abstract}
\thispagestyle{empty}

\section{Introduction} \label{sec:intro}

\vspace{-0.24cm}

Two are the fundamental results in the theory of auctions. First, Vickrey
observed that there is a simple way to run an auction so that social
welfare (efficiency) is maximized: The second-price (Vickrey) auction is
optimally efficient, independently of how bidder valuations are distributed.
However, the whole point of the Vickrey auction is to deliberately
sacrifice auctioneer revenue in order to
achieve efficiency. If auctioneer revenue is to be maximized, Myerson showed
in 1980 that, when the bidders' valuations are distributed independently, a straightforward auction (essentially, a clever reduction to Vickrey's auction via an ingenious transformation of valuations) achieves this.

These two criteria, social welfare and revenue, are arguably of singular and paramount
importance. It is therefore a
pity that they seem to be at loggerheads:  It is not hard to establish that optimizing any one of these two criteria can
be very suboptimal with respect to the other.
In other words, there is a substantial {\em trade-off}
between these two important and natural objectives. {\em What are the various
intermediate (Pareto) points of this trade-off? And can each such point be computed
--- or all such points summarized somehow --- in polynomial time?}
This is the fundamental problem that we consider in this paper.  See Figure~\ref{pareto} (a)
for a graphical illustration.

\vspace{-0.3cm}
\begin{figure}[H]
\includegraphics[scale = 0.5]{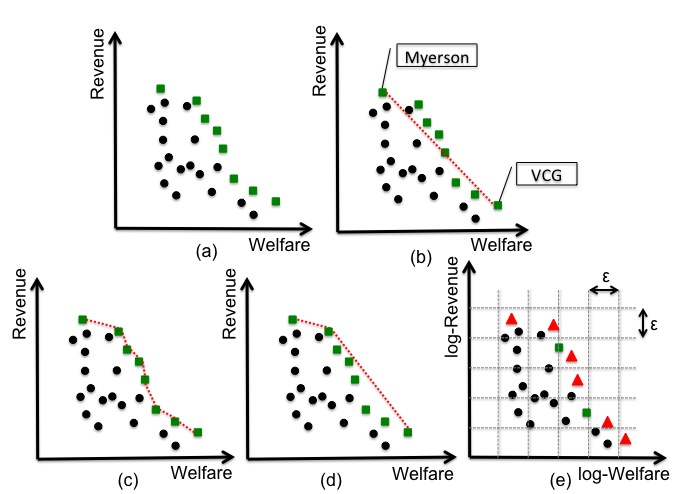}
\caption{The Pareto points of the bi-criterion auction problem are shown as squares (a); the Pareto points may be far off the line connecting the two extremes (b), and may be non-convex (c).  The Pareto points of randomized auctions comprise the upper boundary of the convex closure of the Pareto points (d).  Even though the Pareto set may be exponential in size, for any $\epsilon>0$, there is always a polynomially small set of $\epsilon$-Pareto points, the triangular points in  (e), that is, points that are not dominated by other solutions by more than $\epsilon$ in any dimension.  We study the problem of computing such a set in polynomial time.}\label{pareto}
\end{figure}
\vspace{-0.3cm}

The problem of exploring the revenue/welfare trade-off in auctions turns out to be a rather sophisticated problem, defying several
naive approaches.  One common-sense approach is to simply randomize between the 
optima of the two extremes, Vickrey's and Myerson's auctions.  This can produce very poor results, since it only explores the straight line
joining the two extreme points, which can be very far from the true trade-off (Figure~\ref{pareto} (b)).
A second common-sense approach is the so-called {\em slope search}:  To explore the trade-off space, 
just optimize the objective ``revenue + $\lambda\cdot$ welfare'' for various values of $\lambda >0$.
By modifying Myerson's auction this objective can indeed be optimized efficiently, as 
it was pointed out seven years ago by Likhodedov and Sandholm~\cite{DBLP:conf/sigecom/LikhodedovS04}.
The problem is that the trade-off curve may not be convex (Figure~\ref{pareto} (c)),
and hence the algorithm of ~\cite{DBLP:conf/sigecom/LikhodedovS04} can  
miss vast areas of trade-offs:  

\begin{proposition} \label{prop:convex}
There exist instances with two bidders with monotone hazard rate distributions for which the Pareto curve is not convex; in contrast the Pareto curve is always convex for one bidder with a monotone hazard rate distribution. 
\end{proposition}

The proof is deferred to Appendix~\ref{app:convex}. It follows that the slope search approach of ~\cite{DBLP:conf/sigecom/LikhodedovS04}
is incorrect.  However, the correctness of the slope search approach
is restored if one is willing to settle for randomized mechanisms:   The trade-off space of randomized mechanisms
is always convex (in particular, it is the convex hull of the deterministic mechanisms, (Figure~\ref{pareto} (d)).  It is easy to see (and it had been actually worked out for different purposes already in \cite{MS}) that the optimum randomized mechanism with respect to the metric ``revenue + $\lambda\cdot$ welfare'' is easy to calculate.

\begin{proposition}\label{prop:randomized}
The optimum randomized mechanism for the objective ``revenue + $\lambda\cdot$ welfare'' can be computed in polynomial time.  Hence, any point of the revenue/welfare trade-off for randomized mechanisms can be computed in polynomial time.
\end{proposition}

\subsection{Our results}
{\em In this paper we consider the problem of exploring the revenue/welfare trade-off for {\em deterministic} mechanisms, and show that it is an intractable problem in general, even for two bidders (Theorem~\ref{thm:2p-hard}).}   Comparing with Proposition~\ref{prop:randomized}, this result adds to the recent surge in literature pointing out complexity gaps between randomized and deterministic mechanisms \cite{DBLP:journals/corr/abs-1011-1279,DBLP:journals/corr/abs-1011-2413,DBLP:conf/stoc/DughmiRY11,DBLP:conf/stoc/Dobzinski11}.  Randomized mechanisms are of course a powerful and useful analytical concept, but it is deterministic mechanisms and auctions that we are chiefly interested in. Hence such complexity gaps are meaningful and onerous.   We also show that there are instances for which the set of Pareto optimal mechanisms has exponential size.

On the positive side, we show that the problem can be solved for two bidders, even for correlated valuations (Theorem~\ref{thm:fptas}).  By ``solved'' we mean that any trade-off point can be approximated with arbitrarily high precision in polynomial time in both the input and the precision --- that is to say, by an FPTAS.  It also means (by results in~\cite{DBLP:conf/focs/PapadimitriouY00}) that an approximate summary of the trade-off (the $\epsilon$-Pareto curve), of polynomial size (Figure 1(e)), can be computed in polynomial time.   The derivation of the two-bidders algorithm (see Section~\ref{ssec:DP}) is quite involved.  We first find a pseudo-polynomial dynamic programming algorithm for the problem of finding a mechanism with welfare (resp. revenue) {\em exactly} a given number.  This algorithm is very different from the one in~\cite{DBLP:journals/corr/abs-1011-1279} for optimal auctions in the two bidder case, but it exploits the same feature of the problem, namely its planar nature.  We then recall Theorem 4 of~\cite{DBLP:conf/focs/PapadimitriouY00} (Section~\ref{sec:prelims}) which establishes a connection between such pseudo-polynomial algorithms for the exact problems and FPTAS for the trade-off problem. However, the present problem violates several key assumptions of that theorem, and a custom reduction to the exact problem is needed.

Unfortunately for three or more bidders the above approach no longer works; this is not surprising since, as it was recently shown in~\cite{DBLP:journals/corr/abs-1011-1279}, just maximizing revenue is an APX-hard problem in the correlated case.  The main problem left open in this work is whether there is an FPTAS for three or more bidders with {\em independent} valuation distributions.

We also look at another interesting case of the $n$-bidder problem, in which the valuation distributions have support two.   This case is of  some methodological interest because, in general, $n$-dimensional problems of this sort in mechanism design have not been characterized computationally, because of the difficulty related to the exponential size of the solution sought; binary-valued bidders have served as a first step towards the understanding of auction problems in the past, for example in the study of optimal {\em multi-object} auctions~\cite{RePEc:bla:restud:v:67:y:2000:i:3:p:455-81}. We show that the trade-off problem is in PSPACE and (weakly) NP-hard (Theorem~\ref{thm:many}).   

\subsection{Related work}
Although~\cite{DBLP:conf/sigecom/LikhodedovS04} appears to be the only previous paper explicitly  treating optimal auction design as a multi-objective optimization problem, there has been substantial work in studying the relation of the two objectives. The most prominent paper in the area is that of Bulow and Klemperer~\cite{RePEc:nbr:nberwo:4608} who show that the revenue benefits of adding one  extra bidder and running the efficiency-maximizing auction surpasses those of running the revenue-maximizing auction. In~\cite{DBLP:conf/sigecom/AggarwalGM09} the authors show that for valuations drawn independently from the same monotone hazard rate distribution, an analogous theorem holds for efficiency: by adding $\Theta(\log n)$ extra bidders and running Myerson's auction, one gets at least the efficiency of Vickrey's auction. This paper also shows that for these distributions both the welfare and the revenue ratios between Vickrey and Myerson's auctions are bounded by $1/e$: in our terms this implies that the extreme points of the Pareto curve lie within a constant factor of each other and so constant factor approximations are trivial; we note that no such constant ratios are known for more general distributions (not even for the case of regular distributions), assuming of course that the ratio between all bidders' maximum and minimum valuation is arbitrary. This kind of revenue and welfare ratios are also studied in~\cite{RS07} for keyword auctions (multi-item auctions), and in~\cite{RePEc:eee:gamebe:v:43:y:2003:i:2:p:214-238} for single-item english auctions and valuations drawn from a distribution with bounded support. In~\cite{DBLP:journals/corr/abs-1005-1121} the authors present some tight bounds for the efficiency loss of revenue-optimal mechanisms, which depend on the number of bidders and the size of the support.  Finally, and very recently, \cite{DP} gives simple auctions (in particular, second-price auctions with appropriately chosen reserve prices) that simultaneously guarantee a 20\% fraction of both the optimal revenue and the optimal social welfare, when bidders' valuations are drawn independently from (possibly different) regular distributions: in multiobjective optimization parlance, their auctions belong to the {\em knee} of the Pareto curve. In this work (Section~\ref{sec:fptas}) we provide an algorithm for approximating {\em any} point of the Pareto curve within arbitrary precision, albeit sacrificing the simplicity of the auction format. 

\section{Preliminaries} \label{sec:prelims}
\vspace{-0.24cm}

\subsection{Bayesian Mechanism Design} 

\vspace{-0.2cm}

We are interested in auctioning a single, indivisible item to $n$ bidders. We assume every bidder $i$ has a private valuation $v_i$ for the item and that her valuation is drawn from some discrete probability distribution over support of size $h_i$ with probability density function $f_i(\cdot)$. We use $v_i^k$ and $f_i^k,\, k=1, \ldots, h_i$, to denote the $k$-th smallest element in the support of bidder $i$ and its probability mass respectively.

Formally an auction consists of an allocation rule $x_i(v_1,\ldots,v_n)$, the probability of bidder $i$ getting allocated the item, and a payment rule $p_i(v_1,\ldots,v_n)$ which is the price paid by bidder $i$. In this paper we focus our attention on deterministic mechanisms so that $x_i(\cdot)\in \{0,1\}$. We demand from our auctions to satisfy the two standard constraints of ex-post incentive compatibility (IC) and individual rationality (IR); it is well known~\cite{Nisan2007} that any such  auction has the following special form: if we fix the valuation of all bidders except for bidder $i$, then there is a threshold value $t_i(v_{-i})$, such that bidder $i$ only gets the item for values $v_i\geq t_i(v_{i})$ and pays $t_i(v_{-i})$. In particular one can show that, for the discrete setting and for the objectives of welfare and revenue we are interested in, we can wlog assume that the threshold values $t_i$ of any Pareto optimal auction will always be on the support of bidder $i$.

Relying on the above characterization, we will describe our mechanisms using the concept of an {\em allocation matrix} $A$: a $h_1\times\ldots\times h_n$ matrix where entry $(i_1,\ldots,i_n)$ corresponds to the tuple $(v_1^{i_1},\ldots,v_n^{i_n})$ of bidder's valuations. Each entry takes values from $\{0,1,\ldots,n\}$ indicating which bidder gets allocated the item for the given tuple of valuations, with $0$ indicating that the auctioneer keeps the item.  In order for an allocation matrix to correspond to a valid (ex-post IC and IR) auction a necessary and sufficient condition is the following {\em monotonicity constraint}: if $A[i_1,\ldots,i_j,\ldots,i_n] = j$ then $A[i_1,\ldots,k,\ldots,i_n] = j$ for all $k\geq i_j$. Notice that the payment of the bidder who gets allocated the item can be determined as the least value in his support for which he still gets the item, keeping the values of the other bidders fixed; moreover, when there is only a constant number of bidders, the allocation matrix provides a polynomial representation of an auction. 

\subsection{Multi-Objective Optimization}\label{ssec:mo}

\vspace{-0.2cm}

Trade-offs are present everywhere in life and science --- in fact, one can argue
that optimization theory studies the very special and degenerate case in
which we happen to be interested in only one objective. There is a long research
tradition of {\em multi-objective} or {\em multi-criterion optimization},
developing methodologies for computing the trade-off points (called the
{\em Pareto set}\/) of optimization problems with many objectives, see for
example~\cite{Cli,Ehr,Mit}. However, there is a computational awkwardness
about this problem: Even for simple cases, such as bicriterion shortest paths,
the Pareto set (the set of all undominated feasible solutions) can be exponential,
and thus it can never be polynomially computed. In 2000, Papadimitriou and
Yannakakis \cite{DBLP:conf/focs/PapadimitriouY00} identified a sense in which this is a meaningful
problem: They showed that there is {\em always} a set of solutions of
polynomial size that are {\em approximately} undominated, within arbitrary
precision; a multi-objective problem is considered tractable if such a set can be computed in
polynomial time. Since then, much progress has been made in the algorithmic theory
of multi-objective optimization~\cite{VY,DY1,DY2,FRS09,DDY,CVZ11,diakonikolas}, and much methodology has been developed, some of which has been applied to mechanism design before~\cite{DBLP:conf/soda/GrandoniKLV10}. In this paper we use this methodology for studying Bayesian auctions under the two criteria of expected revenue and social welfare.

\smallskip

\noindent {\bf The {\sc Bi-Criterion Auction} problem.}
We want to design deterministic auctions that perform favorably with respect to
(expected) social welfare, defined as $\sw=\Exp[\sum_i x_i v_i]$ and (expected) revenue, defined as $\rev=\Exp[\sum_ip_i]$. Based on the aforementioned characterization with allocation matrices,
we can view an auction as a feasible solution to a combinatorial problem.
An instance specifies the number $n$ of bidders and for each bidder its distribution on valuations.
The size of the instance is the number of bits needed to represent these distributions. We map solutions
(mechanisms) to points $(x,y)$ in the plane, where we use the $x$-axis for the welfare and the $y$-axis for the revenue. The objective space is the set of such points.

Let $p,q \in \mathbb{R}^2_{+}$. We say that $p$ dominates $q$ if $p \geq q$ (coordinate-wise). We say that $p$ $\eps$-covers
$q$ ($\eps \geq 0$) if $p \ge q / (1+\eps)$. Let $A \subseteq \R^2_+$. The Pareto set of $A$, denoted by $P(A)$, is the subset of
undominated points in $A$ (i.e. $p \in P(A)$ iff $p \in A$ and no other point in $A$ dominates $p$).
We say that $P(A)$ is {\em convex} if it contains no points that are dominated by convex combinations of other points.
Given a set $A \subseteq \R^2_+$ and $\eps>0$, an {\em $\epsilon$-Pareto set} of $A$, denoted by $P_{\epsilon}(A)$, is
a subset of points in $A$ that $\eps$-cover all vectors in $A$.
Given two mechanisms $M, M'$ we define domination between them according to the $2$-vectors of their objective values.
This naturally defines the Pareto set and approximate Pareto sets for our auction setting.

\vspace{-0.06cm}

As shown in~\cite{DBLP:conf/focs/PapadimitriouY00}, for every instance and $\epsilon>0$, there exists an $\epsilon$-Pareto set of polynomial size.
The issue is one of efficient computability. There is a simple necessary and sufficient condition,
which relates the efficient computability of an $\epsilon$-Pareto set to the following {\it GAP
Problem}: given an instance $I$, a (positive rational) $2$-vector $b=(W_0, R_0)$, and a rational $\delta>0$, either return a
mechanism $M$ whose $2$-vector dominates $b$, i.e. $\sw(M) \geq W_0$ and $\rev(M) \geq R_0$,
or report that there does {\em not} exist any mechanism that is better than $b$ by at least a $(1+\delta)$ factor in both coordinates,
i.e. such that $\sw(M) \geq (1+\delta)\cdot W_0$ and $\rev(M) \geq (1+\delta) \cdot R_0$.
There is an FPTAS for constructing an
$\epsilon$-Pareto set iff there is an FPTAS for the GAP Problem~\cite{DBLP:conf/focs/PapadimitriouY00}.

\vspace{-0.1cm}

\begin{remark}
Even though our exposition focuses on discrete distributions, our results easily extend to continuous distributions as well. As in~\cite{DBLP:journals/corr/abs-1011-1279},  given a sufficiently smooth continuous density (say Lipschitz-continuous), whose support lies in a finite interval $[\underline{v},\overline{v}]$,\footnote{This is the standard approach in economics, see for example~\cite{RePEc:nwu:cmsems:362}.} we can appropriately discretize (while preserving the optimal values within $O(\eps)$) and run our algorithms on the discrete approximations.
\end{remark}

\noindent {\bf From exact to bi-criterion.}
We will make essential use of a result from \cite{DBLP:conf/focs/PapadimitriouY00} reducing the multi-objective version of a linear optimization problem $A$ to its exact version:   
Let $A$ be a discrete linear optimization problem whose objective function(s) have {\em non-negative} coefficients. The {\em exact version} of a $A$ is the following problem: Given an instance $x$ of $A$, and a positive rational $C$, is there a feasible solution with objective function value {\em exactly} $C$?
For such problems, a pseudo-polynomial algorithm for the exact version of implies an FPTAS for the multi-objective version:

\begin{theorem}[\cite{DBLP:conf/focs/PapadimitriouY00}]\label{thm:PY2000}
Let $A$ be a {\em linear} multi-objective problem whose objective functions have {\em non-negative} coefficients: 
If there exists a pseudo-polynomial algorithm for the {\em exact version} of $A$,
then there exists an FPTAS for constructing an approximate Pareto curve for $A$.
\end{theorem}

To obtain our main algorithmic result (Theorem~\ref{thm:fptas}), we design a pseudo-polynomial algorithm for the exact version of the {\sc Bi-Criterion Auction} problem and apply Theorem~\ref{thm:PY2000} to deduce the existence of an FPTAS. However, it is not obvious why {\sc Bi-Criterion Auction} satisfies the condition of the theorem, since in the standard representation of the problem as a linear problem, the objective functions typically have negative coefficients. We show however (Lemma~\ref{lem:discrete-case}) that there exists an alternate representation with monotonic linear functions.

\section{The complexity of Pareto optimal auctions}\label{sec:LB}

Our main result in this section is that -- in contrast with randomized auctions -- designing deterministic Pareto optimal auctions under welfare and revenue objectives is an intractable problem; in particular, we show that, even for $2$ bidders~\footnote{Note that for a single bidder, one can enumerate all feasible mechanisms in linear time.} whose distributions are independent and regular, the problem of maximizing one criterion subject to a lower bound on the other is (weakly) NP-hard.

\begin{theorem}\label{thm:2p-hard}
For two bidders with independent regular distributions, it is NP-hard to decide whether there exists an auction with welfare at least $W$ and revenue at least $R$.
\end{theorem}

\begin{proof}[Sketch]
Due to space constraints, we show here the reduction for the exact problem for the welfare objective; quite simple and intuitive, it also captures the main idea in the (significantly more elaborate) proof for the bi-criterion problem (given in Appendix~\ref{app:2p-hard}).

The reduction is from the Partition problem: we are given a set $B = \{b_1, \ldots, b_k \}$ of $k$ positive integers, 
and we wish to determine whether it is possible to partition $B$ into two subsets with equal sum.
We assume that $b_i \geq b_{i+1}$ for all $i$. Consider the rescaled values $b'_i := b_i / (10k \cdot T)$, 
where   $T = \littlesum_{i=1}^k b_i$, and the set $B' =  \{b'_1, \ldots, b'_k \}$. 
It is clear that there exists a partition of $B$ iff there exists a partition of $B'$.

We construct an instance of the auction problem with two bidders whose independent valuations 
$v_r$ (row bidder) and $v_c$ (column bidder) are uniformly distributed over supports of size $k$. 
(To avoid unnecessary clutter in the expressions, we assume w.l.o.g -- by linearity -- that the ``probability mass'' of all elements in the support is equal to $1$, 
as opposed to $1/k$.) The valuation distribution for the row bidder is 
supported on the set $\{1, 2, \ldots, k\}$, while the column bidder's valuation comes from the set 
$\{1+b'_1, 2+b'_2, \ldots,  k+b'_k\}.$ Since $b'_i \geq b'_{i+1}$ and $ \littlesum_{i=1}^k b'_i = 1/(10k)$,
it is straightforward to verify that both distributions are indeed regular (see Appendix~\ref{app:reg}).

The main idea of the proof is this: appropriately {\em isolate} a subset of $2^k$ feasible mechanisms whose welfare values
encode the sum of values $\littlesum_{i \in S} b'_i$ for all possible subsets $S \subseteq [k]$. 
The existence of a mechanism with a specified welfare value would then reveal the existence of a partition.
Formally, we prove that there exists a Partition of $B'$ iff there exists a feasible mechanism $M^{\ast}$ with (expected) welfare
\vspace{-0.2cm}
\begin{equation} \label{eqn:goal}
\sw(M^{\ast}) = (2/3)\cdot(k-1)k(k+1)+ (1/2) \cdot k(k+1) + \littlesum_{i=2}^k (i-1)b'_i + 1/(20k)
\end{equation}

Consider the allocation matrix of a feasible mechanism. Recall that a mechanism is feasible iff its allocation matrix satisfies
the monotonicity constraint. The main claim is that {\em all mechanisms that could potentially satisfy~(\ref{eqn:goal}) 
must allocate the item to the highest bidder, except potentially for the outcomes $(v_r = i, v_c = i+b'_i)$ 
(i.e. the ones corresponding to entries on the secondary diagonal of the matrix) when the item can be allocated to either bidder.}
Denote by $\mathcal{R}$ the aforementioned subclass of mechanisms. The above claim follows from the next lemma, which 
shows that mechanisms in $\mathcal{R}$ maximize welfare (see Appendix~\ref{app:r} for the proof):
\vspace{-0.1cm}
\begin{lemma} \label{lem:r}
We have $ \max_{M \notin \mathcal{R}} {\sw(M)} < \min_{M \in \mathcal{R}} {\sw(M)} < \sw(M^{\ast})$.
\end{lemma}
\vspace{-0.1cm}
To complete the proof,
observe that all $2^k$ mechanisms in $\mathcal{R}$ satisfy monotonicity, hence are feasible. Also note that there is a natural bijection between subsets $S \subseteq [k]$  
and these mechanisms: we include $i$ in $S$ iff on input $(v_r = i, v_c = i+b'_i)$ the item is allocated to the column bidder. Denote by $M(S)$ the mechanism in 
$\mathcal{R}$ corresponding to subset $S$ under this mapping; we will compute the welfare of $M(S)$. 
Note that the contribution of each entry of the allocation matrix (input) to the welfare equals the valuation of the bidder who gets the item for that input.
By the definition of $\mathcal{R}$, for the entries above the secondary diagonal,  the row bidder gets the item (since her valuation
is strictly larger than that of the column bidder -- this is evident since $\max_i b'_i < 1/(10k)$). 
Therefore, the contribution of these entries to the welfare equals $\littlesum_{i=2}^k i(i-1) = (1/3)(k-1)k(k+1)$.
Similarly,  for the entries below the diagonal,  the column bidder gets the item and their contribution to the welfare is 
$\littlesum_{i=2}^k (i+b'_i)(i-1) = (1/3)(k-1)k(k+1)+ \littlesum_{i=2}^k (i-1)b'_i.$ Finally, for the diagonal entries, if $S \subseteq [k]$ is the subset of indices
for which the column bidder gets the item, the welfare contribution is $\littlesum_{i \in S} (i+b'_i) + \littlesum_{i \in [k] \setminus S} i = k(k+1)/2 + \littlesum_{i \in S} b'_i$.
Hence, we have:
\vspace{-0.2cm}
\begin{equation} \label{eqn:have}
\sw(M(S)) = (2/3)\cdot (k-1)k(k+1)+ (1/2)\cdot k(k+1) + \littlesum_{i=2}^k (i-1)b'_i +  \littlesum_{i \in S} b'_i
\end{equation}

Recalling that $\littlesum_{i=1}^k b'_i = 1/(10k)$, (\ref{eqn:goal}) and (\ref{eqn:have}) imply that there exists a partition of $B'$ iff there exists a feasible
mechanism satisfying  (\ref{eqn:goal}). This completes the proof sketch. 
(See Appendix~\ref{app:2p-bic} for the much more elaborate proof of the general case.) \qed
\end{proof}

We can also prove that the size of the Pareto curve can be exponentially large (in other words, the problem of computing the entire curve is exponential even if $P = NP$). 
The construction is given in Appendix~\ref{app:exp}.

\begin{theorem} \label{thm:exp}
There exists a family of two-bidder instances for which the size of the Pareto curve for {\sc Bi-Criterion Auction} grows exponentially.
\end{theorem}

\vspace{-0.3cm}

\section{An FPTAS for 2 bidders} \label{sec:fptas}
\vspace{-0.1cm}
In this section we give our main algorithmic result:
\vspace{-0.1cm}
\begin{theorem}\label{thm:fptas}
For two bidders, there is an FPTAS to approximate the Pareto curve of the {\sc Bi-Criterion Auction} problem, even for arbitrarily correlated distributions.
\end{theorem}

In the proof, we design a pseudo-polynomial algorithm for the exact version of the problem 
(for both the welfare and revenue objectives) and then appeal to Theorem~\ref{thm:PY2000}.  There is a difficulty, however, in showing that the problem satisfies the assumptions of Theorem~\ref{thm:PY2000}, because in the most natural linear representation of the problem, the coefficients for revenue, coinciding with the virtual valuations, may be negative, thus violating the hypothesis of Theorem~\ref{thm:PY2000}.

We use the following alternate representation: Instead of considering the contribution of each entry (bid tuple) of the allocation matrix separately, we consider the revenue and welfare resulting from all the {\em single-bidder mechanisms} (pricings) obtained by fixing the valuation of the other bidder. 

\begin{definition}\label{def:coefs}
 Let $r_1^{i_1,i_2}$ and $w_1^{i_1,i_2}$ be the (contribution to the) revenue and welfare from bidder 1 of the pricing which offers bidder 1 a price of $v_1^{i_1}$ when bidder 2's value is $v_2^{i_2}$:  $r_1^{i_1,i_2}= \sum_{j\geq i_1}v_1^{i_1}\cdot f(j,i_2)$ and $w_1^{i_1,i_2}=\sum_{j\geq i_1}v_1^j\cdot f(j,i_2)$, where $f(\cdot,\cdot)$ is the joint (possibly non-product) valuation distribution. (The quantities $r_2^{i_1,i_2}$ and $w_2^{i_1,i_2}$ are defined analogously.)
\end{definition}

\begin{lemma}\label{lem:discrete-case}
The {\sc Bi-Criterion Auction} problem can be expressed in a way that satisfies the conditions of Theorem~\ref{thm:PY2000}.
\end{lemma}

\begin{proof}
We consider variables $x_{ij},\, y_{ij}$, $i \in [h_1]$, $j \in [h_2]$. The $x_{ij}$'s are defined as follows:
 $x_{ij}=1$ iff $A[i,j]=1$ and $A[i',j]\ne 1$ for all $i'<i$. 
I.e.  $x_{ij}=1$ iff the $(i,j)$-th entry of $A$ is allocated to bidder 1 and, for this fixed value of $j$, $i$ is the smallest index for which bidder 1 gets allocated; 
symmetrically, $y_{ij}=1$ iff $A[i,j]=2$ and $A[i,j']\neq 2$ for all $j'<j$.
It is easy to see that the feasibility constraints are linear in these variables. 
 We can also express the objectives as linear functions with non-negative coefficients as follows:
 \vspace{-0.3cm}
\begin{eqnarray*}
\rev(x,y)&=&\littlesum_{i=1}^{h_1}\littlesum_{j=1}^{h_2}x_{ij} r_1^{i,j}+\littlesum_{i=1}^{h_1}\littlesum_{j=1}^{h_2}y_{ij} r_2^{i,j}\\
\sw(x,y)&=&\littlesum_{i=1}^{h_1}\littlesum_{j=1}^{h_2}x_{ij} w_1^{i,j}+\littlesum_{i=1}^{h_1}\littlesum_{j=1}^{h_2}y_{ij} w_2^{i,j}
\end{eqnarray*}\qed
\end{proof}

\subsection{An algorithm for the exact version of {\sc Bi-Criterion Auction}} \label{ssec:DP}
The main idea behind our algorithm, inspired by the characterization of Lemma~\ref{lem:discrete-case}, is to consider the contribution from each bidder (fixing the value of the other) independently, by going over all (linearly many) single-bidder mechanisms for both bidders. The challenging part is to combine the individual single-bidder mechanisms into a single two-bidder mechanism and to this end we employ dynamic programming:

Assume that both bidders have valuations of support size $h$; the subproblems we consider in our
dynamic program correspond to settings where we condition that the valuation of each bidder is drawn from an upwards closed subset of his original support. Formally, let $M[i,j,W]$ be True iff there exists an auction that uses the valuations $(v_1^i,\ldots,v_1^h)$ and $(v_2^j,\ldots,v_2^h)$
and has welfare exactly $W$. In what follows $N_{i,j}$ is the normalization factor for valuations (jointly) drawn from $(v_1^i,\ldots,v_1^h)$ and $(v_2^j,\ldots,v_2^h)$, namely $N_{i,j} = \sum_{k\geq i, l\geq j}f(v_1^k,v_2^l)$.

\begin{lemma}\label{lem:recurrence}
We can update the quantity $M[i,j,W]$ as follows:
\begin{eqnarray*}
M[i,j,W]=&&\bigvee_{k\geq j}M[i+1,j,\left(W\cdot N_{i,j}-w_2^{i,k}\right)\cdot N_{i+1,j}^{-1}]\\
&\lor& \bigvee_{k\geq i}M[i,j+1,\left(W\cdot N_{i,j}-w_1^{k,j}\right)\cdot N_{i,j+1}^{-1}]\\
&\lor& \bigvee_{\substack{k> i\\ l> j}}M[i+1,j+1,\left(W\cdot N_{i,j}-w_1^{k,j}-w_2^{i,l}\right)\cdot N_{i+1,j+1}^{-1}]
\end{eqnarray*}
\end{lemma}

\begin{proof}
Let $A[i\ldots h,j\ldots h]$ be the allocation matrix of the auction that results from the above update rule, fixing $i$ and $j$. We start by noting that any allocation matrix $A$ can have one of the following four forms:

\begin{itemize}
\item[{\bf F1: }] There exist $i'$ and $j'$ such that $A[i,j']=1$ and $A[i',j]=2$.
\item[{\bf F2: }] There exists $i'$ such that $A[i',j]=2$ but there is no $j'$ such that $A[i,j']=1$.
\item[{\bf F3: }] There exists $j'$ such that $A[i,j']=1$ but there is no $i'$ such that $A[i',j]=2$.
\item[{\bf F4: }] There exist no $i'$ and $j'$ such that $A[i,j']=1$ or $A[i',j]=2$.
\end{itemize}

Because of monotonicity it follows immediately that no allocation matrix of form F1 can be valid, and the other three forms correspond to the three terms of the recurrence; finally note that for any such form, say F2, the first term of the update rule for $M[i,j,W]$ runs over all possible pricings for bidder 1 (keeping the value of bidder 2 at $v_2^j$) and checks whether they induce the required welfare. \qed
\end{proof}

We omit the straightforwards details of how the above recurrence can be efficiently implemented as a pseudo-polynomial dynamic programming algorithm.
The algorithm for deciding whether there exists an auction with revenue exactly $R$ is identical to the above by simply replacing $R$ (the revenue target value) for $W$ and $r_j^{i_1,i_2}$ for $w_j^{i_1,i_2}$. 

\section{The case of  $n$ bidders} \label{sec:many-bidders}
When the number $n$ of bidders is part of the input, the allocation matrix is no longer a polynomially succinct representation of a mechanism. In fact, it is by no means clear whether  {\sc Bi-Criterion Auction} is even in $NP$ in this case: we next show that for the case of $n$ binary bidders, the problem is $NP$-hard and in $PSPACE$:

\begin{theorem}\label{thm:many}
For $n$ binary-valued bidders {\sc Bi-Criterion Auction} is (weakly) NP-hard and in PSPACE.
\end{theorem}

\begin{proof}[Sketch]
For simplicity, we prove both results for the exact version of the problem for welfare; the bi-objective case follows by a straightforward but tedious generalization.

The NP-hardness reduction is from Partition. Let $B=\{b_1,\ldots,b_k\}$ be a set of positive rationals; we can assume by rescaling that $\littlesum_{i=1}^k b_i = 1/100$. 
We construct an instance of the auction problem as follows: there are $k$ bidders, with uniform distributions (again we will assume unit masses for simplicity) over the following supports $\{l_i,h_i\}, i=1\ldots n$, where $l_i<h_i$. We set $l_i = b_i$ and demand that $\{h_i\}_{i=1,\ldots, n}$ forms a super-increasing sequence (i.e. $h_{i+1}>\sum_{j=1}^ih_j$), with $h_1 > \max_i b_i$.
The claim is that there exists a partition of $B$ iff there exists an auction with welfare equal to $\littlesum_{i=1}^k h_i +(1/2) \littlesum_{i=1}^k b_i$. To see this notice that -- since the sequence $\{h_i\}_{i=1\ldots n}$ is super-increasing -- any mechanism with the above welfare value must must allocate to bidder $i$ for {\em exactly} one valuation tuple $(v_i,v_{-i})$ where $v_i = h_i$; the corresponding contribution to the welfare from this case is $h_i$. Monotonicity then implies that this auction can allocate to bidder $i$ for {\em at most} one valuation tuple $(v_i,v_{-i})$ where $v_i = l_i$; the corresponding contribution to the welfare from this case is $b_i$. We therefore get a bijection between subsets of $B$ and mechanisms, by including an element $b_i$ in the set $S$ iff bidder $i$ gets allocated the item for some valuation tuple $(v_i,v_{-i})$ where $v_i=l_i$, and the claim follows. 

For the PSPACE upper bound,  we start by noting that the problem of computing an auction with welfare (or revenue) {\em exactly} $W$, can be formulated as the problem of computing a matching of weight exactly $W$ in a particular type of bipartite graphs (first pointed out in~\cite{DBLP:journals/corr/abs-1011-2413}, see Appendix~\ref{app:many}) with a number of nodes that is exponential in the number of bidders. 
The {\sc Exact Matching} problem is known to be solvable in RNC~\cite{DBLP:conf/stoc/MulmuleyVV87}; since our input provides an exponentially succinct representation of the constructed graph, we are interested in the so-called {\em succinct version} of the problem~\cite{DBLP:journals/iandc/GalperinW83,DBLP:journals/iandc/PapadimitriouY86}. By standard techniques, the succinct version of {\sc Exact Matching} in our setting is solvable in PSPACE, and the theorem follows. \qed
\end{proof}

We conjecture the above upper bound to be tight (i.e. the problem is actually PSPACE-complete) even for $n$ bidders with arbitrary supports.

\section{Open Questions} \label{sec:conclusions}

\vspace{-0.2cm}
Is there is an FPTAS for 3 bidders?  We conjecture that there is, and in fact for any constant number of bidders.  Of course, the approach of our FPTAS for 2 bidders cannot be generalized, since it works for the correlated case, which is APX-complete for 3 or more bidders.  We have derived two different dynamic programming-based PTAS's for the uncorrelated problem, but so far, despite a hopeful outlook, we have failed to generalize them to 3 bidders.  Finally, we conjecture that for $n$ bidders the problem is significantly harder, namely PSPACE-complete and inapproximable.

On a different note, it would be interesting to see if we can get better approximations for some special types of distributions; we give one such type of result in Appendix~\ref{app:convex}. Are there improved approximation guarantees for more general kinds of distributions and $n$ bidders? 

\bibliographystyle{abbrv}

\bibliography{refs}

\newpage
\appendix

\section*{Appendix}

We provide here proofs that did not appear in the main body, due to space limitations.

\section{Proof of Proposition~\ref{prop:convex}} \label{app:convex}

Let $F$ denote the cumulative distribution function of $f$. We say a distribution satisfies the {\em monotone hazard rate condition} iff the ratio $\frac{1-F(t)}{f(t)}$ is non-increasing; notice in particular that any binary-valued distribution satisfies the monotone hazard rate condition.

We start with a simple example with 2 bidders for which the Pareto curve is not convex, while the bidders' valuations are drawn independently from two {\em non-identical} distributions of support 2; this is presented in Figure~\ref{example3}.

\begin{figure}[h]
\centering
\includegraphics[scale = 0.5]{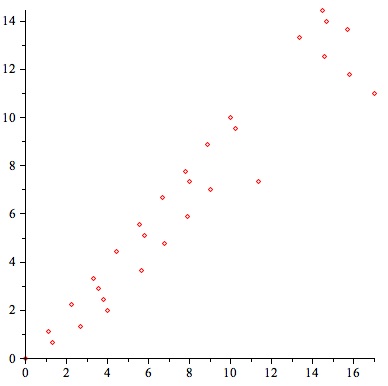}
\caption{The objective space of all two-player auctions when $(v_1^1,v_1^2)=(11,20),(v_2^1,v_2^2)=(2,5)$ and $(f_i^1,f_i^2)=(1/3,2/3)$ for $i=1,2$. }\label{example3}
\end{figure}

On the positive side we show that the Pareto curve is convex for a single bidder with valuation drawn from a monotone hazard rate distribution. Since in the discrete case one can simply enumerate the set of all feasible auctions in linear time anyway, this result is of interest only in the continuous case.

Let $M(r)$ be the single-player mechanism (pricing) that makes a take-it-or-leave-it offer of $r$ to the player; notice that in this case the Pareto curve is a mono-parametric curve on the plane where $x=\sw[M(r)]$ and $y=\rev[M(r)]$. We next show that this monoparametric curve is in fact convex; the following is a necessary and sufficient condition for (local) convexity of mono-parametric (continuous) curves:

\[ \left| \begin{array}{cc}
x'(r) & x''(r)\\
y'(r) & y''(r) \end{array} \right| =x'(r)y''(r)-y'(r)x''(r)\geq0\]

Substituting $x=\sw[M(r)]=\int_r^1vf(v)\,dv$ and $y=\rev[M(r)]=r\int_r^1f(v)\,dv$ and doing the algebra we get the following necessary and sufficient condition:
\begin{equation}\label{convexity}
r(f(r))^2+\int_r^1f(v)\,dv\cdot\left[f(r)+rf'(r)\right]\geq0
\end{equation}
By definition, for monotone hazard rate distributions the ratio $\frac{1-F(r)}{f(r)}$ is a non-increasing function of $r$; taking derivatives\footnote{We make the analytically convenient assumption that $f$ is differentiable here.} we get that for these distributions it must hold that:
\begin{equation}\label{mhr}
(1-F(r))f'(r)+(f(r))^2\geq0
\end{equation}
By substituting $f'(r)$ from (\ref{mhr}) into the LHS of (\ref{convexity}) we get that it is indeed $\geq 0$.

An interesting open question is whether the property of convex Pareto curves extends to 2 or more bidders with valuations distributed {\em identically} according to some monotone hazard rate distribution.

\section{Omitted Proofs from Theorem~\ref{thm:2p-hard}} \label{app:2p-hard}

\subsection{Proof of Regularity} \label{app:reg}

We define the virtual valuation of a bidder with valuation $v_i$,  taking values from $\{v_i^1,\ldots,v_i^k\}$ with probabilities $\{f_i^1,\ldots,f_i^k\}$, as follows: 
$$\phi_i^j = v_i^j - (v_i^{j+1}-v_i^j)\frac{f_i^{j+1}+\ldots+f_i^k}{f_i^j}$$
Substituting for the setting in hand, we get that the virtual valuation of the row player is $\phi_r^j = 2j-k$, while the virtual valuation of the column player is:
$$\phi_c^j = j+b'_j - (b'_{j+1}-b'_j+1)(k-j)$$
A distribution is called regular iff $\phi_i^j\leq \phi_i^{j+1}$; it follows immediately that the row player's distribution is regular, while for the second player we need that
$$j+1+b'_{j+1} - (b'_{j+2}-b'_{j+1}+1)(k-j-1) \geq j+b'_j - (b'_{j+1}-b'_j+1)(k-j)$$
and since $b'_{j+1}\geq b'_{j+2}$ it suffices that
$$j+1+b'_{j+1} - (k-j-1) \geq j+b'_j - (b'_{j+1}-b'_j+1)(k-j)$$
Rearranging terms and doing some calculations we get that it suffices to have
$b'_{j+1}\geq b'_j-2/k$, which follows from the fact that $\sum_{i=1}^k b_i' = 1/(10k)$.

\subsection{Proof of Lemma~\ref{lem:r}} \label{app:r}

Recall that this lemma implies that only mechanisms in $\mathcal{R}$ can potentially satisfy (\ref{eqn:goal}).
To prove it we proceed as follows: Consider a partition of the allocation matrix $A$ into three subsets:
(i) the subset of the matrix above the secondary diagonal, (ii) the subset below the diagonal and (iii) the diagonal itself.
The gist of the proof is this: The contribution to the welfare for subsets (i) and (ii) is maximized for mechanisms in $\mathcal{R}$.
The welfare contribution from (i) and (ii) for {\em any} other mechanism (i.e. not in $\mathcal{R}$) is strictly smaller by a 
quantity sufficiently large that outweighs any effects on the welfare from subset (iii).

Let us first compute $\min_{M \in \mathcal{R}} \sw(M)$, the minimum welfare of a mechanism in $\mathcal{R}$.
It is easy to see that the welfare minimizing mechanism is the one that assigns the item to the row player for all entries in the diagonal.
That is, it corresponds to $S = \emptyset$ in the bijection defined in the body of the proof, hence we have
\[ \min_{M \in \mathcal{R}} \sw(M) = (2/3)\cdot (k-1)k(k+1)+ (1/2)\cdot k(k+1) + \littlesum_{i=2}^k (i-1)b'_i.\]
So, we obtain the second inequality of the lemma.
To bound from above $\max_{M \not\in \mathcal{R}} \sw(M)$ we consider three cases:
Consider first the subset of the allocation matrix above the diagonal. If any entry of this subset is allocated to the column
player, then it is not hard to see that this would lower the welfare value by at least $1- \max_i b'_i  \geq 1-1/(10k) \geq 0.9$.
Similarly, if any entry is not allocated al all (i.e. the auctioneer keeps the item), this would cost us at least $1$.
For the subset below the diagonal the situation is analogous; if an entry is allocated to the row player, this costs us
at least $1$, same if an entry is not allocated at all. This decrease in the value of the welfare cannot be compensated
by the diagonal entries; indeed, if all such entries are allocated to either player, contribution to the welfare lies in $[k(k+1)/2, k(k+1)/2+1/(10k)]$
(an interval of length $1/(10k) \leq 1/10.$) As a consequence, any mechanism that disagrees with $\mathcal{R}$ either below
or above the diagonal has welfare strictly smaller than $\min_{M \in \mathcal{R}} \sw(M)$. Now consider a mechanism
that agrees with $\mathcal{R}$ except potentially at the diagonal. Note that a non-allocated entry of the diagonal costs at least $1$,
and again this cannot be compensated by the $1/10$ potential contribution of the column player. This completes the proof.

\subsection{Proof of Theorem~\ref{thm:2p-hard}} \label{app:2p-bic}

The high-level idea is similar as the proof presented in the body of the paper (for the exact version of the welfare objective) 
but the details are more elaborate. At a high-level, the difficulty is that the two objective functions (welfare, revenue)
depend on each other in a subtle way. Thus, a more complicated construction is required to ``decouple''  these two criteria. 
(It is not hard to see that the construction presented for the exact version fails for the bi-objective problem.) Very roughly,
the reduction ends up using non-uniform distributions on larger and carefully selected supports.

\medskip

As before, our reduction is from Partition. We start with a set $A=\{a_1,\ldots,a_k\}$ of positive numbers (rescaled so that they sum to a sufficiently small positive constant) and we want to decide whether there exists a partition of this set. We will construct an instance of the auction problem with $2$ players and distributions of support size $2k+1$. Before presenting the actual instance we first give some intuition behind the construction.

Similarly, our goal  is to establish a bijection between an appropriate subset of feasible mechanisms and subsets $S$ of  $[k]$; since the number of feasible mechanisms greatly exceeds that of subsets of $A$, we have to limit our attention to a subset of feasible mechanisms. To that end, we are going to appropriately pick the target values for welfare and revenue, so that the only relevant auctions in our reduction will be those that allocate the item to player $2$ (column player) for entries above the diagonal (i.e. $(v_1^i,v_2^j)$ with $j>i$) and to player $1$ (row player) for entries below the diagonal (i.e. $(v_1^i,v_2^j)$ with $j<i$). We will also exclude the possibility of not allocating the item across the diagonal entries of the allocation matrix, so that the only relevant auctions are the $2^{2k+1}$ different auctions that allocate to either player $1$ or $2$ across the diagonal, all of which respect monotonicity and are therefore feasible. Call this subset of mechanisms $\mathcal{R}$. We will then use the $i$-th {\em odd} entry of the diagonal to encode the decision of including or not the $i$-the element of $A$ in the set $S$: we shall include element $i$ iff player 1 gets allocated for entry $(2i-1,2i-1)$ of the allocation matrix.

The first step is therefore to ensure that the only relevant mechanisms are the ones with the above property. To this end we ask that the following relation between the players' valuations holds:
\begin{equation}\label{property1}
v_1^i < v_2^i < v_1^{i+1}<v_2^{i+1}<v_1^{i+2}, \text{ for }i=1\ldots2k-1
\end{equation}
Relation~(\ref{property1}) implies that the social welfare from the entries on top and below the diagonal is maximized by a mechanism that allocates to player 2 on top of the diagonal and to player 1 below the diagonal; therefore by setting a sufficiently high welfare target $W$ in our reduction we will be able to guarantee that the only relevant auctions will have this format.

More specifically, the distributions are defined as follows (where $\eps>0$ a sufficiently small parameter):

The (unnormalized) probabilities of the two players are:

\[ f_j^i = \left\{ \begin{array}{ll}
         1 & \mbox{if $i$ is odd};\\
        \epsilon & \mbox{if $i$ is even}.\end{array} \right. \]
for both players $j=1,2$, where $\epsilon$ is some small constant to be determined later.
(The point of the small probability elements, is to achieve the desired decoupling between welfare and revenue; it may be convenient for the reader to think of $\eps$ as if it was $0$. In the course of the proof, we will provide a sufficient upper bound on its magnitude.)

The values of player $1$ are:
\[ v_1^i = \left\{ \begin{array}{ll}
         i+a_{\frac{i+1}{2}} & \mbox{for $i\in\{1,3,\ldots,2k-1$\}};\\
        i+a_{\frac{i}{2}}\left(1+\frac{4}{(2k-i+2)(1+\epsilon)}\right) & \mbox{for $i\in\{2,4,\ldots,2k$\}};\\
        2k+1 & \mbox{for $i=2k+1$}.\end{array} \right. \]

The values of player $2$ are:
\[ v_2^i = \left\{ \begin{array}{ll}
         i&\mbox{for $i\in\{1,3,\ldots,2k-1$\}};\\
        i&  \mbox{for $i\in\{2,4,\ldots,2k$\}};\\
        2k+1 & \mbox{for $i=2k+1$}.\end{array} \right. \]

We note that there are $3$ different scales of numbers in the reduction. The values of the elements in the support (big scale), the magnitudes of the elements of $A$ (medium scale), and the magnitude of $\eps$ (small scale). 

Recall we would like to make the sum of welfare and revenue remain constant across all mechanisms in $\mathcal{R}$; by doing so we can ensure that whenever a mechanism achieves the target welfare and revenue values, the relations will in fact hold with equality, allowing us to encode an instance of Partition. To achieve that, we impose an even stronger requirement: In particular, consider the following entries of the allocation matrix: $(v_1^i,v_2^j), j=i\ldots 2k+1$ and $(v_1^j,v_2^i), j=i\ldots 2k+1$, where $i$ is an odd number. Assuming our auction has the format discussed above, entries $(v_1^i,v_2^j), j=i+1\ldots 2k+1$ are allocated to player 2, entries $(v_1^j,v_2^i), j=i+1\ldots 2k+1$ are allocated to player 1, and we are left to decide which player to allocate entry $(v_1^i,v_2^i)$ to. Now let $\sw_j^i$ (resp. $\rev_j^i$), where $i$ is odd and $j\in\{1,2\}$, denote the welfare (resp. revenue) that results from the aforementioned entries if we allocate entry $(v_1^i,v_2^i)$ to player $j$. The stronger requirement that we impose is that $\sw_1^i+\rev_1^i = \sw_2^i+\rev_2^i$ for all odd $i$. To see what this entails we next write the expressions for $\sw_j^i$ and $\rev_j^i$:
\begin{eqnarray*}
\sw_1^i &=& v_1^i + \sum_{j=\frac{i+1}{2}}^kv_1^{2j+1}+ \sum_{j=\frac{i+1}{2}}^kv_2^{2j+1}+\epsilon\cdot\left(\sum_{j=\frac{i+1}{2}}^kv_1^{2j}+\sum_{j=\frac{i+1}{2}}^kv_2^{2j} \right)\\
\rev_1^i &=& v_1^i\left(\frac{2k-i+1}{2}(1+\epsilon)+1\right) + v_2^{i+1}\left(\frac{2k-i-1}{2}+ \frac{2k-i+1}{2}\epsilon+1\right)\\
\sw_2^i &=& \sum_{j=\frac{i+1}{2}}^kv_1^{2j+1}+v_2^i + \sum_{j=\frac{i+1}{2}}^kv_2^{2j+1}+\epsilon\cdot\left(\sum_{j=\frac{i+1}{2}}^kv_1^{2j}+\sum_{j=\frac{i+1}{2}}^kv_2^{2j} \right)\\
\rev_2^i &=& v_1^{i+1}\left(\frac{2k-i-1}{2}+ \frac{2k-i+1}{2}\epsilon+1\right)+v_2^i\left(\frac{2k-i+1}{2}(1+\epsilon)+1\right)\\
\end{eqnarray*}
Notice that $\sw_1^i-\sw_2^i = v_1^i-v_2^i$. In order to have $\sw_1^i+\rev_1^i = \sw_2^i+\rev_2^i$ we ask that:
\begin{equation}\label{property2}
\rev_1^i-\rev_2^i = v_2^i-v_1^i
\end{equation}

The only difficulty in satisfying the relation above, is that equation~(\ref{property2}) necessarily imposes some additional constraints on the values $v_1^{i+1}, v_2^{i+1}$; we get around this by using a support of roughly twice the size of $A$, and using only half of the points in the support to encode the elements of $A$; the remaining points are assigned a very small probability, so that they have a negligible effect on the overall welfare and revenue. It is now easy to verify that the aforementioned choice of distributions for the two players satisfies properties~(\ref{property1}) (since $a_i$ are assumed to be sufficiently small) and~(\ref{property2}) above.

For the aforementioned choice of values $v_j^i$ the social welfare contributions now become:
$$\sw_1^i = v_1^i+X_i = i+a_{\frac{i+1}{2}} +X_i \text{ and } \sw_2^i = v_2^i+X_i = i+X_i$$
for some $X_i$ whose exact value is irrelevant (and can be derived from the expressions above); analogously for revenue we have:
$$\rev_1^i = v_2^i+Y_i=i+Y_i \text{ and }\rev_2^i = v_1^i+Y_i =  i+a_{\frac{i+1}{2}} +Y_i$$ for some $Y_i$. We therefore have $\sw_1^i+\rev_1^i = \sw_2^i+\rev_2^i = 2i + X_i+Y_i + a_{\frac{i+1}{2}}$ and we have thus ensured that all mechanisms with the property of allocating to player 2 on top of the diagonal and to player 1 below the diagonal have a sum of (total) revenue and welfare that can be upper-bounded as follows:
$$\sw+\rev\leq\sum_{\text{odd }i} (2i+X_i+Y_i+a_{\frac{i+1}{2}}) +v_1^{2k+1}+v_2^{2k+1}+\epsilon\cdot 2n^2(2k+1)$$
where the last term is an upper bound on the contribution in revenue and welfare of the even rows and columns (where we took into account that the maximum contribution of any entry is at most the maximum value appearing in the support of any player, namely $2k+1$). We next fix the value of $\epsilon$ so that the quantity
$$\epsilon \cdot 2n^2(2k+1)$$ is smaller than the accuracy used in the rational numbers $a_i$.
Note that this can always be done with an $\eps$ that has polynomially many bits -- since the $a_i$'s are by assumption rational numbers with polynomially many bits.

We are now ready to argue that there exists a partition of $A$ iff there exists an auction with:
\begin{equation}\label{eq:eq1}
\sw\geq \sum_{\text{odd }i} (i+X_i) + 2k+1+\frac{1}{2}\sum_{i=1}^k a_i \,\,\,\,\,\text{and}\,\,\,\,\,\,\, \rev\geq \sum_{\text{odd }i} (i+Y_i) +2k+1+ \frac{1}{2}\sum_{i=1}^k a_i
\end{equation}

Given any partition $S$ of $A$, we can turn it into a mechanism with the above welfare and revenue guarantees by allocating to player 2 on top of the diagonal, player 1 below the diagonal and allocating to player 1 for entries $(2i-1,2i-1), i=1\ldots k$, for all $i$ s.t. $a_i\in S$; the even entries on the diagonal, as well as the entry $(2k+1,2k+1)$ can be allocated to either player.

Conversely, given a mechanism with welfare and revenue as above we can get a partition of $A$. To see how, first notice that because of property~(\ref{property1}) above (and because $a_i$ are much smaller) the only mechanisms that can achieve a social welfare of at least $\sum_{\text{odd }i}i+X_i$ and revenue of at least $\sum_{\text{odd }i}i+Y_i$ are those that allocate to player 2 above the diagonal, player 1 below the diagonal, and always allocate to either player 1 or player 2 on the diagonal. In the discussion above we established that for those mechanisms it holds that:
\begin{equation}\label{eq:eq2}
\sw+\rev \leq \left(\sum_{\text{odd }i} (i+X_i) + 2k+1+\frac{1}{2}\sum_{i=1}^k a_i\right) + \left(\sum_{\text{odd }i} (i+Y_i) +2k+1+ \frac{1}{2}\sum_{i=1}^k a_i\right) + \epsilon\cdot 2n^2(2k+1)
\end{equation}
By our choice of $\epsilon$ and inequalities~(\ref{eq:eq1}) and~(\ref{eq:eq2}) it follows that the inequalities in~(\ref{eq:eq1}) must hold with equality; we then get a partition by including in $S$ all elements $i$ for which $(2i-1,2i-1)$ is allocated to player 1. This completes the proof. 

\medskip

\section{Proof of Theorem~\ref{thm:exp}} \label{app:exp}

The construction is similar to the reduction for the exact problem in Theorem~\ref{thm:2p-hard}. We will construct a $2$ player auction and we will argue
that there exists an appropriate {\em subset} of the Pareto curve with exponential size. 

We describe an instance with $2$ players, both with uniform distributions over the following supports of size $k$: Player 1 has values in $\{1+a_1,2+a_2,\ldots,k+a_k\}$, and player 2 has values in $\{1,2,\ldots,k\}$, with all $a_i<<1$; the exact value of $a_i$ will be determined later. Assume player 1 is on the horizontal axis of the two-player allocation matrix and consider a mechanism that allocates the item to player 2 for all entries above the diagonal (i.e. $(v_1^i,v_2^j)$ with $j>i$), to player 1 for all entries below the diagonal (i.e. $(v_1^i,v_2^j)$ with $j<i$), and to either player on the diagonal (see Table~\ref{exp-Pareto}); such a mechanism can be concisely described through the diagonal entries of its allocation matrix. In what follows we write $\sw(A[v_1^1,v_2^1],\ldots,A[v_1^k,v_2^k])$ and $\rev(A[v_1^1,v_2^1],\ldots,A[v_1^k,v_2^k])$ to denote the welfare and revenue respectively of this mechanism. We note that this subset of feasible mechanisms maximizes the welfare over all feasible mechanisms; hence, it suffices to show that the Pareto set {\em of this subset} is exponential. In fact, we will choose the $a_i$'s appropriately so that {\em all} these mechanisms are undominated.

Our goal is to pick values $a_1,\ldots,a_k$ such that all $2^k$ different mechanisms of the above type will be Pareto optimal. To do that we observe that under some mild conditions on the $a_i$, satisfied by picking for example $a_i=3^{i-1}$ (and normalizing so that the normalized sum is small e.g. $< 1/1000$), we can impose orderings on the welfares and revenues of those $2^k$ mechanisms that go in opposite directions, i.e. one mechanism has larger revenue than another  iff it has smaller welfare. To this end we make the following two claims, which can be verified by explicitly writing down the expressions for revenue and welfare and doing some elementary calculations:

\paragraph{Claim 1: } If $a_i>0$ and $a_i < \frac{n-i}{n-i+1}a_{i+1}$ for all $i$, it holds that:
\begin{enumerate}
\item $\sw(\xi_1,\ldots,\xi_{i-1},1,\xi_{i+1},\ldots,\xi_k)> \sw(\xi_1,\ldots,\xi_{i-1},2,\xi_{i+1},\ldots,\xi_k)$
\item $\rev(\xi_1,\ldots,\xi_{i-1},1,\xi_{i+1},\ldots,\xi_k)< \rev(\xi_1,\ldots,\xi_{i-1},2,\xi_{i+1},\ldots,\xi_k)$
\end{enumerate}
for any $\xi_i\in \{1,2\}$.

\paragraph{Claim 2: } If $\sum_{j=1}^{i-1}a_j< a_i$ and $a_i < \frac{n-i}{2(n-i+1)}a_{i+1}$ for all $i$, it holds that:
\begin{enumerate}
\item $\sw(1,\ldots,1,2,\xi_{i+1},\ldots,\xi_k)< \sw(2,\ldots,2,1,\psi_{i+1},\ldots,\psi_k)$
\item $\rev(1,\ldots,1,2,\xi_{i+1},\ldots,\xi_k)> \rev(2,\ldots,2,1,\psi_{i+1},\ldots,\psi_k)$
\end{enumerate}
for any $\xi_i,\psi_i\in \{1,2\}$; note that in general we may have $\psi_i\neq\xi_i$.

Intuitively Claim 1 says that switching any 1 into a 2 on any entry of the diagonal has the effect of decreasing the welfare while increasing the revenue; Claim 2 on the other hand says that the (negative) effect that a 2 on the diagonal has on the welfare, is bigger for 2's that are placed in higher positions --and in analogy for revenue.

Using the above two claims one can now prove that for any two mechanisms $M_1=(\xi_1,\ldots,\xi_k)$ and $M_2=(\psi_1,\ldots,\psi_k)$, it holds that $\sw(M_1)>\sw(M_2)$ iff $\rev(M_1)<\rev(M_2)$, and therefore all the $2^k$ mechanisms are Pareto optimal. We convey the idea by means of the following example, for $k=5$; consider the two mechanisms $M_1=(1,2,1,1,1)$ and $M_2 = (1,1,2,2,1)$. We then have:
$$\sw(M_1)=\sw(1,2,1,1,1)>\sw(2,2,2,1,1)>\sw(1,1,1,2,1)>\sw(1,1,2,2,1)=\sw(M_2)$$
with the inequalities above following from Claim 1.1, 2.1 and 1.1 respectively; in complete analogy we can show that $\rev(M_1)<\rev(M_2)$.

\begin{table}[h]
\begin{center}
  \begin{tabular}{|c||c|c|c|c|c| }
    \hline
    $v_2^k=k$&2&2&2&\ldots&1 or 2\\\hline
    $v_2^{k-1}=k-1$&2&2&\ldots&1 or 2&1\\\hline
    $\vdots$&$\vdots$&$\vdots$&$\reflectbox{$\ddots$}$&$\vdots$&$\vdots$\\\hline
    $v_2^2=2$&2&1 or 2&\ldots&1&1\\\hline
    $v_2^1=1$&1 or 2&1&\ldots&1&1\\\hline\hline
    &$v_1^1 = 1+a_1$&$v_1^2 = 2+a_2$&\ldots&$v_1^{k-1} = k-1+a_{k-1}$&$v_1^k = k+a_k$\\\hline
  \end{tabular}
\end{center}
\caption{An instance with an exponential size Pareto set}\label{exp-Pareto}
\end{table}

\section{Omitted Details from Theorem~\ref{thm:many}} \label{app:many}

The graph is the following: We assume that each bidder has two values $\{v_i^1,v_i^2\}$, with $v_i^2>v_i^1$, and for each of the $2^n$ valuation tuples we create a node in the graph labeled by this tuple $(v_i)_i$. We connect two such nodes if their labels differ in exactly one coordinate, say the $i$-th one; the weight of this edge is $\prod_{j\neq i} f_j^{k_j}\cdot(f_i^1v_i^1+f_i^2v_i^2)$ for welfare and $\prod_{j\neq i} f_j^{k_j}\cdot v_i^1(f_i^1+f_i^2)$ for revenue. We also introduce a set of dummy nodes: for every node with label $(v_i)_i$, where $|\{i|v_i=v_i^2\}|=k$, we introduce $k$ dummy nodes and associate each one of them with the bidder $i$ for whom $v_i=v_i^2$. We then add an edge between this node and all its dummy nodes; the weight of the edge connecting to the dummy node of the $i$-th player is $\prod_{j\neq i}f_j^{k_j}\cdot f_i^2v_i^2$ both for welfare and revenue.

It is easy to verify that every matching in the above graph corresponds to a deterministic truthful mechanism as follows: for each bid vector, consider the corresponding node in the graph. If it is not matched to any other node, then allocate nothing; if it is matched to a dummy node, then allocate the item to the bidder that is associated with this dummy node; otherwise, it is matched to another non-dummy node and these two nodes differ in exactly one coordinate, say the $i$-th, in which case we allocate to the $i$-th bidder. It is easy to check that the resulting mechanism is both feasible and monotone (IC and IR). Moreover, the welfare (or revenue depending on the kind of weights used) is equal to the weight of the matching.

\end{document}